\documentclass[11pt]{article}
\usepackage[applemac]{inputenc}
\usepackage{cite}
\usepackage{amssymb}
\usepackage{amsmath}
\usepackage{relate}
\usepackage{fullpage}
\usepackage{graphicx}

\newtheorem{lemma}{Lemma}
\newtheorem{theorem}{Theorem}
\newtheorem{corollary}{Corollary}
\newtheorem{definition}{Definition}

\newcommand{\qed}{\hfill\ensuremath{\Box}\medskip\\\noindent}
\newenvironment{proof}{\noindent\emph{Proof. }}{\qed}

\newcommand{\aps}{\ensuremath{\sqsubseteq}}
\newcommand{\rhop}{\ensuremath{\gamma}}
\newcommand{\aass}{arc-preserving split}

\newcommand{\lightdepth}{\ensuremath{\mathrm{lightdepth}}}
\newcommand{\spaces}{\ensuremath{\mathrm{spaces}}}
\newcommand{\size}{\ensuremath{\mathrm{size}}}
\newcommand{\rank}{\ensuremath{\textsc{rank}}}
\newcommand{\select}{\ensuremath{\textsc{select}}}

\title{Fast Arc-Annotated Subsequence Matching in Linear Space\thanks{An extended abstract of this paper appeared in proceedings of the 36th International Conference on Current Trends 
in Theory and Practice of Computer Science.}}

\author{Philip Bille\thanks{Supported by the Danish Agency for Science, Technology, and Innovation.} \\ \texttt{phbi@imm.dtu.dk}  \and Inge Li G{\o}rtz\thanks{Corresponding author. Address: Technical University of Denmark, Informatics and Mathematical Modelling, Building 322, Office 124, DK-2800 Kongens Lyngby, Denmark.
Phone: +45 45 25 36 73. Fax: +45 45 93 00 74} \\ \texttt{ilg@imm.dtu.dk}}

\date{}

\begin{document}
\maketitle

\begin{abstract}
An arc-annotated string is a string of characters, called bases, augmented with a set of pairs, called arcs, each connecting two bases.  Given arc-annotated strings $P$ and $Q$ the arc-preserving subsequence problem is to determine if $P$ can be obtained from $Q$ by deleting bases from $Q$. Whenever a base is deleted any arc with an endpoint in that base is also deleted. Arc-annotated strings where the arcs are ``nested'' are a natural model of RNA molecules that captures both the primary and secondary structure of these. The arc-preserving subsequence problem for nested arc-annotated strings is basic primitive for investigating the function of RNA molecules. Gramm et al. [ACM Trans. Algorithms 2006] gave an algorithm for this problem using $O(nm)$ time and space, where $m$ and $n$ are the lengths of $P$ and $Q$, respectively. In this paper we present a new algorithm using $O(nm)$ time and $O(n + m)$ space, thereby matching the previous time bound while significantly reducing the space from a quadratic term to linear. This is essential to process large RNA molecules where the space is likely to be a bottleneck. To obtain our result we introduce several novel ideas which may be of independent interest for related problems on arc-annotated strings.
\end{abstract}

\section{Introduction}
An \emph{arc-annotated string} $S$ is a string augmented with an \emph{arc set} $A_S$. Each character in $S$ is called a \emph{base} and the arc set $A_S$ is a set of pairs of positions in $S$ connecting two distinct bases. We say that $S$ is a \emph{nested arc-annotated string} if no two arcs in $A_S$ share an endpoint and no two arcs cross each other, i.e., for all $(i_l, i_r), (i_l', i_r') \in A_S$ we have that $i_l < i_l' < i_r$ iff $i_l < i_r' < i_r$. Given arc-annotated strings $P$ and $Q$ we say that $P$ is a \emph{arc-preserving subsequence} (APS) of $Q$, denoted $P \sqsubseteq Q$, if $P$ can be obtained from $Q$ by deleting $0$ or more bases from $Q$. Whenever a base is deleted any arc with an endpoint in that base is also deleted. The \emph{arc-preserving subsequence problem}  (APS) is to determine if $P \sqsubseteq Q$. If $P$ and $Q$ are both nested arc-annotated strings we refer to the problem as the \emph{nested arc-preserving subsequence problem} (NAPS). Fig.~\ref{fig:treestructure}(a) shows an example of nested arc-annotated strings. 

\begin{figure}[t] 
  \centering \includegraphics[scale=.5]{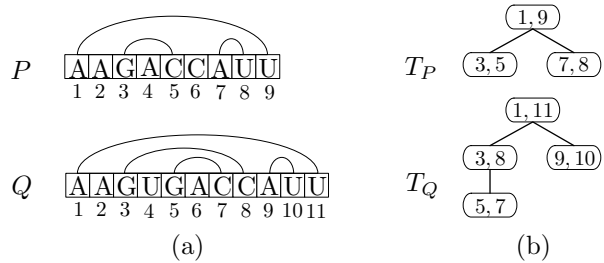}
  \caption{(a) Nested arc-annotated strings $P$ and $Q$. Here, $P$ and $Q$ contain arcs connecting their first and last bases. (b) The corresponding trees $T_P$ and $T_Q$ induced by the arcs.}
  \label{fig:treestructure}
\end{figure}

Ribonucleic acid (RNA) molecules are often modeled as nested arc-annotated strings. Here, the string consists of bases from the $4$-letter alphabet $\{\textrm{A}, \textrm{U}, \textrm{C}, \textrm{G}\}$, called the \emph{primary structure}, and an arc set consisting of pairings between bases, called the \emph{secondary structure}. The secondary structure of RNA is central for many biological functions and it is more preserved in evolution than the primary structure~\cite{JCW2000, TKFHVP2000, HMBP1996, FKPVH1999, BBFFFHH2007, BB2007}. NAPS is a simple primitive for comparing the secondary structure of the RNA molecules. Furthermore, it may also serve as the subroutine in algorithms for the more general \emph{longest common arc-preserving subsequence problem} (LAPCS), see Gramm et al.~\cite{GJN2006}.

Building on earlier work in a related model of RNA molecules by Vialette~\cite{Vialette2004}, Gramm  et al.~\cite{GJN2006} introduced and gave an algorithm for NAPS using $O(nm)$ time and space, where $m$ and $n$ are the lengths of $P$ and $Q$, respectively. Kida~\cite{Kida2006} presented an experimental study of this algorithm and Damaschke~\cite{Damaschke2006} considered a special restricted case of the problem.

\subsection{Results}
We assume a standard unit-cost RAM model with word size $\Theta(\log n)$ and a standard instruction set including arithmetic operations, bitwise boolean 
operations, and shifts. The space complexity is the number of words used by the algorithm. All of the previous results are in same model of computation. Throughout the paper $P$ and $Q$ are nested arc-annotated strings of lengths $m$ and $n$, respectively. In this paper we present a new algorithm with the following complexities.
\begin{theorem}\label{thm:main}
Given nested arc-annotated strings $P$ and $Q$ of lengths $m$ and $n$, respectively, we can solve the nested arc-preserving subsequence problem in time $O(nm)$ and space $O(n+m)$.
\end{theorem}
Hence, we match the running time of the currently fastest known algorithm and at the same time we improve the space from $O(nm)$ to $O(n + m)$. This space improvement is critical for processing large RNA molecules. In particular, an algorithm using $O(nm)$ space quickly becomes infeasible, even for moderate sizes of RNA molecules, due to costly accesses to external memory.  An algorithm using $O(m+ n)$ space is much more scalable and allows us to handle significantly larger RNA molecules. Furthermore, we note that 
obtaining an algorithm using $O(nm)$ time and $o(nm)$ space is mentioned as an open problem in Gramm~et~al.~\cite{GJN2006}.

Compared to the previous work by Gramm et al.~\cite{GJN2006} our algorithm is not only more space-efficient but also simpler. Our algorithm is based on a single unified dynamic programming recurrence, whereas the algorithm by Gramm et al. requires computing and tabulating auxiliary information in multiple phases mixed with dynamic programming. Our approach allows us to better expose the features of NAPS and is essential for obtaining a linear space algorithm.

\subsection{Techniques}
As mentioned above, our algorithm is based on a new dynamic programming recurrence. Essentially, the recursion expresses for any pair of substrings $P'$ and $Q'$ of $P$ and $Q$, respectively, the longest prefix of $P'$ which is an arc-preserving subsequence of $Q'$ in term of smaller substrings of $P'$ and $Q'$. We combine several new ideas with  well-known techniques to convert our recurrence into an efficient algorithm.

First, we organize the dynamic programming recurrence into \emph{$\Gamma$ sequences}. A $\Gamma$ sequence for a given substring $Q'$ of $Q$ is a simple $O(m)$ space representation of the longest arc-preserving subsequences of each prefix of $P$ in $Q'$. We show how to efficiently manipulate $\Gamma$ sequences to get new $\Gamma$ sequences using a small set of simple operations, called the \emph{primitive operations}. Secondly, we organize the computation of $\Gamma$ sequences using a recursive algorithm that traverses the tree structure of the arcs in $Q$. The algorithm computes the $\Gamma$ sequence for each arc in $Q$ using the primitive operations. To avoid storing too many $\Gamma$ sequences during the traversal we direct the computation according to the well-known \emph{heavy-path decomposition} of the tree. This leads to an algorithm that stores at most $O(\log |A_Q|)$ $\Gamma$ sequences. Since each $\Gamma$ sequence uses $O(m)$ space the total space becomes $O(m \log |A_Q| + n)$.

Finally, to achieve linear space we exploit a structural property of $\Gamma$ sequences to compress them efficiently. We obtain a new representation of $\Gamma$ sequences that only requires $O(m)$ \emph{bits}. Plugging in the new representation into our algorithm the total space becomes $O(n + m)$ as desired. However, the resulting algorithm requires many costly compressions and decompressions of $\Gamma$ sequences at each arc in the traversal. As a practical and more elegant solution we show how to augment the compressed representation of $\Gamma$ sequences using standard \emph{rank/select indices} to obtain constant time random access to elements in $\Gamma$ sequences. This allows us to compress each $\Gamma$ sequence only once and avoid decompression entirely without affecting the complexity of the algorithm.

\subsection{Related Work}
Arc-annotated strings are a natural model of RNA molecules that captures both the primary and secondary structure of these. Consequently, a wide range of pattern matching problems for them have been studied,  see e.g.,~\cite{Evans1999, BMR1995, GCJW2002, AGGN2004, BT2006, GJN2006, BLMTW2009}. Among these, NAPS is one of the most basic problems.

The NAPS problem generalizes the \emph{tree inclusion problem} for ordered trees~\cite{Chen1998, KM1995, BG2005}. Here, the goal is to determine if a tree can be obtained from another tree by deleting nodes. This is equivalent to NAPS where all bases in both strings have an incident arc. The authors have shown how to solve the tree inclusion problem in time $O(nm/\log n + n\log n)$ and space $O(n+m)$~\cite{BG2005}. Compared to our current result for NAPS the space complexity is the same but the time complexity for tree inclusion is a factor $O(\log n)$ better for most values of $m$ and $n$. Though our obtained complexities for the tree inclusion problem and NAPS are very similar, the ideas and techniques behind the results differ significantly. While the definition of the two problems seems very similar it appears that the more general NAPS is significantly more complicated. We leave it as an interesting research direction to determine the precise relationship between NAPS and the tree inclusion problem.

Several generalizations of NAPS have also been studied relaxing the requirement that arcs should be nested~\cite{GJN2006, Evans1999, BFRV2005}. In nearly all cases the resulting problem becomes NP-complete.

 \subsection{Outline}
 In Sec.~\ref{sec:preliminaries} we give some preliminaries and define our 
 notation. Sec.~\ref{sec:recurrence} contains our dynamic programming recurrence. In Sec.~\ref{sec:algorithm} we present our main algorithm 
 achieving $O(m\log |A_Q| + n)$ space. Finally, in Sec.~\ref{sec:linearspace} we show how to compress the $\Gamma$ sequences stored by our algorithm to obtain $O(n + m)$ space.

\section{Preliminaries and Notation}\label{sec:preliminaries}
Let $S$ be an arc-annotated string with arc set $A_S$. The length of $S$ is the number of bases in $S$ and is denoted $|S|$. We will assume that our input strings $P$ and $Q$ have the arcs $(1, |P|)$ and $(1, |Q|)$, respectively. If this is not the case we may always add additional connected bases to the start and end of $P$ and $Q$ without affecting the solution or complexity of the problem. We do this only to ensure that the nesting of the arcs form a tree (rather than a forest) which simplifies the presentation of our algorithm.

The \emph{arc-annotated substring} $S[i_1, i_2]$,  $1 \leq i_1, i_2 \leq |S|$, is the string of bases starting at $i_1$ and ending at $i_2$.  The arc set associated with $S[i_1, i_2]$ is the subset of $A_S$ of arcs with both endpoints in $[i_1, i_2]$. We define $S[i_1] = S[i_1, i_1]$ and $S[i_1, i_2] = \epsilon$ (the empty string) if $i_1 > i_2$. Note the arc set of an arc-annotated string of length $\leq 1$ is also empty. A \emph{split} of $S$ is a partition of $S$ into two substrings $S[1, i]$ and $S[i+1, |S|]$, for some $i$, $0 \leq i \leq |S|$. The split is an \emph{arc-preserving split} if no arcs in $A_S$ cross $i$, i.e., all arcs either have both endpoints in $S[1, i]$ or $S[i+1, |S|]$. We say that the index $i$ \emph{induces} a (arc-preserving) split of $S$. 

\begin{figure}[t] 
  \centering \includegraphics[scale=.5]{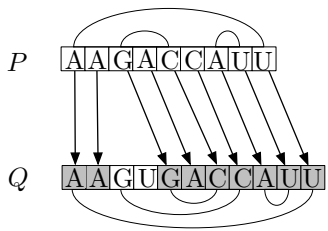}
  \caption{An embedding of $P$ in $Q$. $f(1)=1, f(2)=2$ and $f(j)=j+2$ for $j=3,4,\ldots,9$.  $A_P=\{(1,9),(3,5),(7,8)\}$ and $A_Q=\{(1,11),(3,8),(5,7),(9,10)\}$. We have $(f(1),f(9)) =(1,11) \in A_Q$, $(f(3),f(5))=(5,7)\in A_Q$, and $(f(7),f(8))=(9,10)\in A_Q$.}
  \label{fig:embedding}
\end{figure}
An \emph{embedding} of $P$ in $Q$ is an injective function $f: \{1,\ldots, m\} \rightarrow \{1, \ldots, n\}$ such that
\begin{enumerate}
\item  for all $j\in\{1,\ldots, m\}$, $P[j]= Q[f(j)]$. (base match condition)
\item for all indices $j_l,j_r \in \{1,\ldots, m\}$, $(j_l,j_r) \in A_P \Leftrightarrow (f(j_l),f(j_r))\in A_Q$. (arc match condition)
\item for all $i,j\in\{1,\ldots, m\}$, $i <j \Leftrightarrow f(i) < f(j)$. (order condition) 
\end{enumerate}
If $f(j)=i$ we say that $j$ is \emph{matched} to $i$ in the embedding. From the definition of arc-preserving subsequences we have that $P \sqsubseteq Q$ iff there is an embedding of $P$ in $Q$. Figure~\ref{fig:embedding} gives an example of an embedding.

\section{The Dynamic Programming Recurrence}\label{sec:recurrence}

In this section we give our dynamic programming recurrence for NAPS. Essentially, the recursion expresses for any pair of substrings $P'$ and $Q'$ of $P$ and $Q$, respectively, the longest prefix of $P'$ which is an arc-preserving subsequence of $Q'$ in terms of smaller substrings of $P'$ and $Q'$.

We show the following key properties of \aass s. 
%
\begin{figure}[t] 
  \centering \includegraphics[scale=.5]{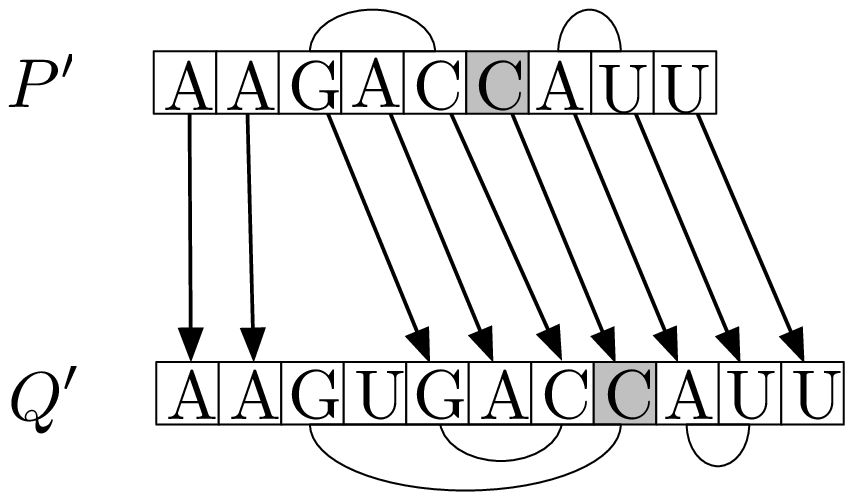} \includegraphics[scale=.5]{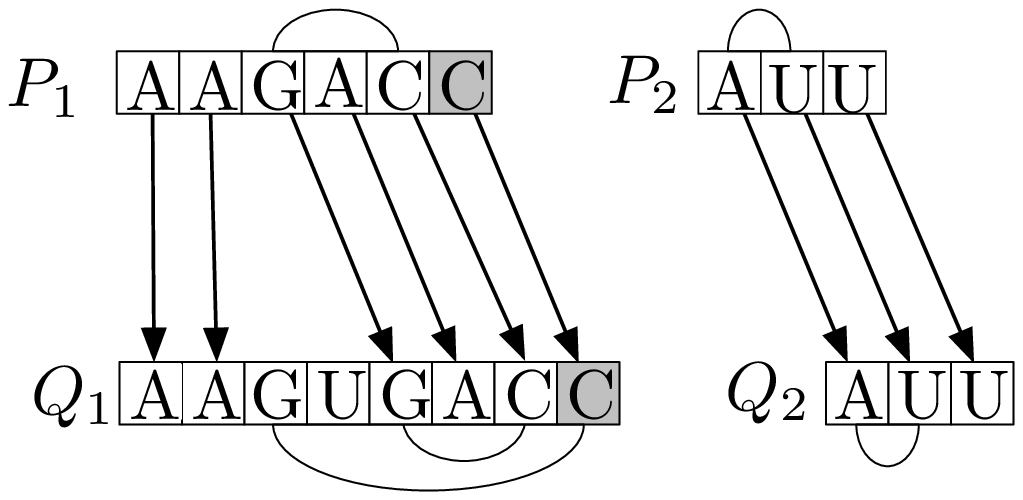}
  \caption{Illustation of the Splitting Lemma. Index $i=8$ induces an \aass\ of $Q'$ with $Q_1=Q'[1,8]$ and $Q_2=Q'[9,11]$. Index $j=6$ induces an \aass\ of $P'$ with $P_1=P'[1,6]$ and $P_2=P'[7,9]$ such that $P_1 \aps Q_1$ and $P_2 \aps Q_2$.}
  \label{fig:embedding}
\end{figure}

\begin{lemma}[Splitting Lemma]\label{lem:split}
Let $P'$ and $Q'$ be arc-annotated substrings of $P$ and $Q$, respectively, and let $(Q_1,Q_2)$ be any \aass\ of $Q'$.
\begin{itemize}
\item[(i)] 
If $P' \aps Q'$ then there exists an  \aass\  $(P_1,P_2)$ of $P'$ such that $P_1 \aps Q_1$ and $P_2 \aps Q_2$. 

\item[(ii)] Let $(P_1,P_2)$ be an \aass\ of $P'$. Then
$P_1 \aps Q_1 \textrm{ and }  P_2 \aps Q_2 \Rightarrow P' \aps Q'$.  
\end{itemize}
\end{lemma}
\begin{proof}
Let $P'=P[j_1,j_2]$, $Q'=Q[i_1,i_2]$, $Q_1=Q[i_1,i]$, and $Q_2=Q[i+1,i_2]$.

To prove 
(i), let $f$ be an embedding of $P'$ in $Q'$ (such an embedding exists since $P'\aps Q'$). Let $j$ be the largest index such that $f(j) \in [i_1,i]$, i.e., $f(j)$ is a base in $Q_1$. It follows that $f_1:\{j_1,\ldots, j\} \rightarrow \{i_1,\ldots, i\}$, where $f_1(x)=f(x)$, is an embedding of $P[j_1,j]$ in $Q_1$, and thus $P[j_1,j] \aps Q_1$. Similarly,  $f_2:\{j+1,\ldots, j_2\} \rightarrow \{i+1,\ldots, i_2\}$, where $f_2(x)=f(x)$,  is an embedding of $P[j+1, j_2]$ in $Q_2$, and therefore $P[j+1,j_2] \aps Q_2$. We have now shown that there exists a split $(P_1,P_2)$, with $P_1=P[j_1,j]$ and $P_2=P[j+1,j_2]$, such that $P_1 \aps Q_1$ and $P_2\aps Q_2$. It remains to show that this is an \emph{arc-preserving} split, i.e., that there are no arcs from 
 $P[j_1,j]$ to $P[j+1,j_2]$ in $A_{P'}$. For contradiction assume that there exists an arc $(j_l,j_r)$ with $j_l \in [j_1,j]$ and $j_r \in [j+1,j_2]$. By the definition of $f_1$ and $f_2$ we have $f(j_l) = f_1(j_l) \in [i_1,i]$ and $f(j_r)=f_2(j_r) \in [i+1,i_2]$. Since $f$ is an embedding of $P'$ in $Q'$ it follows from the arc match condition that there is an arc $(f(j_l),f(j_r))\in A_{Q'}$. But this contradicts the fact that $(Q_1,Q_2)$ is an \aass\ of $Q'$. Thus there can be no such arc and
it follows that $j$ induces an \aass\ of $P'$.

To prove (ii), assume $P_1 \aps Q_1$ and  $P_2 \aps Q_2$. Let $f_1$ be the embedding of $P_1$ in $Q[i_1, i]$, let $f_2$ be the embedding of $P_2$ in $Q[i+1, i_2]$, and let $j$ be the index such that $P_1=P[j_1,j]$. 
We will show that the embedding $f:\{j_1,\ldots, j_2\} \rightarrow \{i_1,\ldots, i_2\}$, $$f(x) = \begin{cases}f_1(x) & x\in [j_1,j] \\ f_2(x) & x\in[j+1,j_2] \end{cases}$$  
 is an embedding of $P'$ in $Q'$.  
 
We have that $f$ satisfies the base match condition and the order condition. It remains to show that it satisfies the arc match condition. 
 Since $j$ induces an \aass\ of $P'$ we have $A_{P'} = A_{P_1} \cup A_{P_2}$. Let $(j_l,j_r)$ be an arc in $A_{P'}$. If $(j_l,j_r)\in A_{P_1}$, then $(f(j_l),f(j_r))=(f_1(j_l),f_1(j_r))\in A_{Q_1}$, since $f_1$ is an embedding of $P_1$ in $Q_1$. Similarly, if  $(j_l,j_r)\in A_{P_2}$, then $(f(j_l),f(j_r))=(f_2(j_l),f_2(j_r))\in A_{Q_2}$. Thus $(j_l,j_r)\in A_{P'} \Rightarrow (f(j_l),f(j_r)) \in A_{Q'}$. 
By the same kind of argument it follows that since $i$ induces an \aass\ of $Q'$, $ (f(j_l),f(j_r)) \in A_{Q'}  \Rightarrow (j_l,j_r)\in A_{P'}$.
\end{proof}
We now define $\rhop$, which we will use to give our dynamic programming recurrence for NAPS.

\begin{definition}
For $1 \leq j_l \leq m$, $l\in\{1,2\}$, and $1 \leq i_1 \leq i_2\leq n$, define $\rhop(j_1, j_2, i_1, i_2)$ to be the largest integer $k$ such that $P[j_1, k] \sqsubseteq Q[i_1, i_2]$ and $k$ induces an \aass\ of $P[j_1,j_2]$.
\end{definition}
It follows that $\rhop(1,m,1,n)=m$ if and only if $P\aps Q$.


The Splitting Lemma gives us 
a very useful  property of $\rhop$:
The requirement  that $k$ induces an \aass\ of $P[j_1,j_2]$ in the definition of $\rhop$ implies that if there exists an embedding $f$ of $P[k+1,j_2]$ in $Q[i_2,i]$ for some $i$ then by the Splitting Lemma the embedding of $P[j_1,k]$ in $Q[i_1,i_2]$ (which exists by the definition of $\rhop$) can be extended with $f$ to get an embedding of $P[j_1,j_2]$ in $Q[i_1,i]$.
This would not be true if we dropped the requirement that $k$ induces an \aass\ of $P[j_1,j_2]$. Formally,
\begin{corollary}\label{cor:split}
Let $i$ be an index inducing an \aass\ of $Q[i_1,i_2]$.
Then, $$\rhop(j_1,j_2,i_1,i_2) = \rhop(\rhop(j_1,j_2,i_1,i)+1,j_2,i+1,i_2)\;.$$
\end{corollary}
\begin{proof} 
 Let $k=\rhop(j_1,j_2,i_1,i_2)$, $j= \rhop(j_1,j_2, i_1,i)$, and  $k'=\rhop(j+1,j_2,i+1,i_2)$. 
We want to show that $k'=k$. We will first show that $k'\leq k$ by using the second part of the Splitting Lemma. Next we show that $k' \geq k$ by using the first part of the splitting lemma.

Let $Q_1=Q[i_1,i]$ and $Q_2=[i+1,i_2]$. By the definition of $\rhop$, $j$ and $k'$ we have $$P[j_1,j]\aps Q_1 \textrm{ and  } P[j+1,k'] \aps Q_2\;.$$
Since $P[j_1,k] \aps Q[i_1,i_2]$ and $j+1 \geq j_1$ we have,
$$P[j+1,k] \aps Q[i_1,i_2]\;.$$
By the definition of $\rhop$, index $j$ induces an \aass\ of $P[j_1,j_2]$. Since $k'\leq j_2$ index $j$ also induces an \aass\ of $P[j_1,k']$.
By the Splitting Lemma (ii) we have $P[j_1,k']\aps Q[i_1,i_2]$. Since $k'$ induces an \aass\ of $P[j+1,j_2]$ and $j$ induces an \aass\ of $P[j_1,j_2]$ we have that $k'$ induces an \aass\ of $P[j_1,j_2]$. This, together with $P[j_1,k']\aps Q[i_1,i_2]$, implies $k' \leq k$, since by definition of $\rhop$, $k$ is the largest index inducing an \aass\ of $P[j_1,j_2]$ such that $P[j_1,k]\aps Q[i_1,i_2]$.

To show that $k \leq k'$ we use the first part of the Splitting Lemma. Since 
$i$ induces an \aass\ of $Q[i_1,i_2]$, by the Splitting Lemma (i) there exists a $j'$ such that $P[j_1,j']\aps Q_1$ and $P[j'+1,k] \aps Q_2$ and $j'$ induces an \aass\ of $P[j_1,j_2]$. Let $j^*$ be the largest such $j'$. We will show that $j^*=\rhop(j_1,j_2,i_1,i)=j$. This implies $P[j+1,k]=P[j^*+1,k] \aps Q_2$ and thus $k\leq k'$ since $k'$ is the \emph{largest} integer such that $P[j+1,k'] \aps Q_2$.

By definition $j$ is the largest integer inducing an \aass\ of $P[j_1,j_2]$ such that $P[j_1,j]\aps Q_1$ and thus $j^* \leq j$. But this implies $P[j+1,k] \aps P[j^*+1,k] \aps Q_2$. Thus $j=j^*$.
\end{proof}


Intuitively, the corollary says that to compute the largest prefix of $P$ that can be embedded in $Q$ we can greedily match the bases and right endpoints of arcs of $P$ as much to the left in $Q$ as possible.
The dynamic programming recurrence for $\rhop$ is as follows. The intuition behind the recurrence is that it corresponds to computing the leftmost embedding of $P[j_1,j_2]$ in $Q[i_1,i_2]$. 
%
\begin{align*}
\intertext{{\bf Base cases}. $\rhop(j_1,j_2,i_1,i_2)$ is equal to}
& 
\begin{cases}
	j_1-1 & \text{if }  j_1>j_2, \qquad \qquad \qquad \qquad\qquad \qquad\qquad \qquad\qquad\qquad \qquad \qquad  \;\;(1) \\
      	j_1 & \text{if } i_1=i_2 \text{ and } P[j_1]=Q[i_1]  \text{ and }   \\
	& \quad  (j_1,j_r)\not\in A_P \text{ for all } j_r\leq j_2 ,  \; \quad \qquad \qquad \qquad \qquad \qquad \qquad \qquad (2a) \\
   	j_1-1 &  \text{if } i_1=i_2 \text{ and (} P[j_1]\neq Q[i_1] \text{ or} \\
	& \quad (j_1,j_r)\in A_P \text{ for some }j_r\leq j_2 \text{).} \quad \qquad \qquad \qquad \qquad \qquad \qquad \,\;\;\; (2b)
\end{cases}  \\
\intertext{{\bf Recursive cases}. $i_1<i_2$ and $j_1 \leq j_2$.}
\intertext{If $(i_1,i_r)\not\in A_Q$ for all  $i_r\leq i_2$ then $\rhop(j_1,j_2,i_1,i_2)$ is equal to}
&
\begin{cases}
 	\rhop(j_1+1,j_2,i_1+1,i_2) & \text{if }  (j_1,j_r)\not\in A_P \text{ for all } j_r\leq j_2  \text{ and } P[j_1]=Q[i_1], \; \quad (3) \\
   	\rhop(j_1,j_2,i_1+1,i_2) & \text{if }  (j_1,j_r)\in A_P \text{ for some } j_r\leq j_2 \text{ or } P[j_1] \neq Q[i_1],   \quad (4) \\
\end{cases}\\
\intertext{If $(i_1,i_r)\in A_Q$  for some  $i_r < i_2$, then $\rhop(j_1,j_2,i_1,i_2)$ is equal to} 
& \rhop(\rhop(j_1,j_2,i_1,i_r)+1,j_2,i_r+1,i_2) \qquad\qquad  \qquad \qquad \qquad \qquad \qquad\qquad \qquad \quad \; \;(5) \\
\intertext{If $(i_1,i_2)\in A_Q$  then $\rhop(j_1,j_2,i_1,i_2)$ is equal to} 
&
\begin{cases}
	\max\{\rhop(j_1,j_2,i_1+1,i_2), \\ 
	\qquad \quad \rhop(j_1,j_2,i_1,i_2-1)\} & \text{if }  (j_1,j_r)\not\in A_P \text{ for all } j_r\leq j_2, \qquad \quad \qquad \quad \;\;\;\, (6) \\
	\rhop(j_1,j_2, i_1+1,i_2) & \text{if }  (j_1,j_r)\in A_P \text{ for some } j_r\leq j_2, \qquad \qquad   \qquad(7) \\ 
	&\qquad \text{and } P[j_1] \neq Q[i_1] \text{ or } P[j_r] \neq Q[i_2],\\
	\max\{\phi,\rhop(j_1,j_2,i_1+1,i_2)\}  & \text{if } (j_1,j_r)\in A_P \text{ for some } j_r\leq j_2, \qquad \qquad \qquad(8) \\ 
	& \qquad P[j_1] = Q[i_1] \text{ and } P[j_r] = Q[i_2],\\
\end{cases} \\
\intertext{where } 
\phi = &\begin{cases}
j_r & \text{if } \rhop(j_1+1,j_r-1, i_1+1,i_2-1)  =j_r -1 \\
j_1-1 
& \text{otherwise.}
\end{cases}
\end{align*}
The cases are visualized in Fig.~\ref{fig:recc}. 

\begin{figure}[t] 
  \centering \includegraphics[scale=.55]{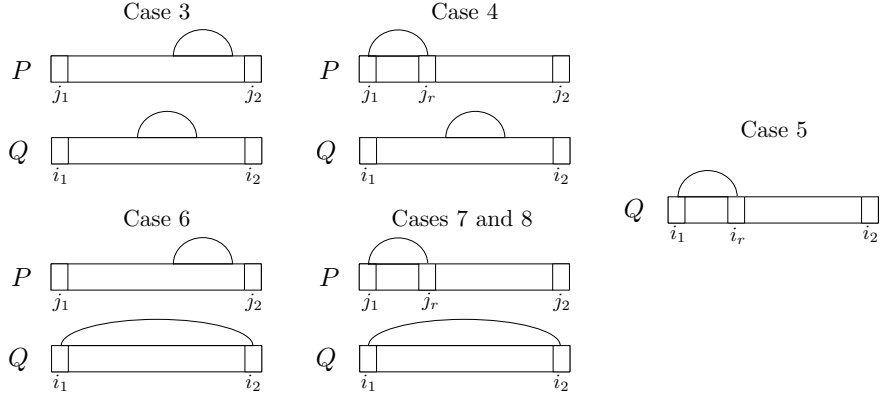}
  \caption{The main cases from the recurrence relation. Case (3): Neither $P$ or $Q$ starts with an arc.  Case (4): $P$ starts with an arc, $Q$ does not. Case (5): $Q$ starts with an arc not spanning $Q$. We split $Q$ after the arc and compute \rhop\ first in the first half and then continue the computation in the other. Case (6): $Q$ starts with an arc, $P$ does not. Case (7)-(8): Both $P$ and $Q$ starts with an arc.}
  \label{fig:recc}
\end{figure}

%
The base cases $(1)-(2)$ cover the cases where  $P[j_1,j_2]$ is the empty string ($j_2 > j_1$) or $Q[i_1,i_2]$ is a single base ($i_1=i_2$). 
Let $k=\rhop(j_1,j_2,i_1,i_2)$.
Case $(3)$ and $(5)$ follows directly from Corollary~\ref{cor:split}. In case $(4)$ and $(7)$ the base $Q[i_1]$ cannot be part of an embedding of $P[j_1,k]$ in $Q[i_1,i_2]$ and thus $\rhop(j_1, j_2, i_1,i_2)=\rhop(j_1, j_2, i_1+1,i_2)$. In case $(6)$ either $Q[i_1]$ or $Q[i_2]$, but not both, can be part of an embedding of $P[j_1,k]$ in $Q[i_1,i_2]$. Thus, $\rhop(j_1, j_2, i_1,i_2)=\max\{\rhop(j_1, j_2, i_1,i_2-1),\rhop(j_1, j_2, i_1+1,i_2)\}$. Case $(8)$ is the most complicated one. Both $Q[i_1,i_2]$ and $P[j_1,j_2]$ start with an arc and the bases of the arcs match.  An embedding of $P[j_1,k]$ into $Q[i_1,i_2]$ either 
(i) matches the two arcs, (ii) matches the arc $(j_1,j_r)$ and the rest of $P[j_1,k]$ in $Q[i_1+1,i_2]$ or (iii) matches nothing ($k=j_1-1$). In case (ii) $\rhop(j_1, j_2, i_1,i_2)=\rhop(j_1, j_2, i_1+1,i_2)$. Case (i) requires that $P[j_1+1, j_r-1] \aps Q[i_1+1,i_2-1]$. We express this in the recurrence by using an auxiliary function $\phi$ which is $j_r$ if $\rhop(j_1+1, j_r-1, i_1+1,i_2-1)= j_r-1$ and $j_1-1$ otherwise, since in the last case the arc $(j_1,j_r)$ cannot be matched to the arc $(i_1,i_2)$.
Since we want the largest match we take the maximum of the two cases (i) and (ii) (case (iii) is covered by these two).

\paragraph{Relation to the Gramm et al.\ algorithm} The recurrence relation is not very different from the one implicit in the algorithm by Gramm  et al.~\cite{GJN2006}. Most of the single cases are the same. The main difference is that they intermix the description of the algorithm and the recurrence. And where we have the requirement that $\rhop(j_1,j_2,i_1,i_2)$ induces an \aass\ of $P[j_1,j_2]$ they instead specify a specific order in which to calculate the recurrence and save some auxiliary information during the computation. Thus our definition of $\rhop$ gives us the possibility to state the recurrence relation independently of the algorithm.



\section{The Algorithm}\label{sec:algorithm}
We now present an algorithm to solve NAPS in $O(nm)$  time and $O(m \log |A_Q| + n)$ space. In the next section we show how to further reduce the space to $O(n + m)$ to get Theorem~\ref{thm:main}. The result relies on a well-known path decomposition for trees applied to arc-annotated strings combined with a new idea to organize the dynamic programming recurrence computation. We present these in Sections~\ref{sec:heavypath} and~\ref{sec:gammasequence} before giving the algorithm in Section~\ref{sec:subalgorithm}.

\subsection{Heavy-Path Decomposition of Arc-Annotated Sequences}\label{sec:heavypath}
Let $S$ be a nested arc-annotated string containing the arc $(1, |S|)$ (recall that we assume that both $P$ and $Q$ have this arc).  The arcs in $A_S$ induce a rooted and ordered tree $T_S$ rooted at the arc $(1, |S|)$ as shown in Fig.~\ref{fig:treestructure}(b). We use standard tree terminology for the relationship between arcs in $T_S$.  Let $(i_l, i_r)$ be an arc in $A_S$. The \emph{depth} of $(i_l, i_r)$ is the number of edges on the path  from $(i_l, i_r)$ to the root in $T_S$. An arc with no children is a leaf arc and otherwise an internal arc. Define $T_S(i_l, i_r)$ to be the subtree of $T_S$ rooted at $(i_l, i_r)$ and let $\size(i_l, i_r)$ be the number of arcs in $T_S(i_l, i_r)$. Note that $\size(1, |S|) = |A_S|$. If $(i_l', i_r')$ is an arc in $T_S(i_l, i_r)$, then $(i_l, i_r)$ is an ancestor of $(i_l', i_r')$ (note that $(i_l, i_r)$ is an ancestor of itself).  If $(i_l, i_r)$ is an ancestor of $(i_l', i_r')$, then $(i_l', i_r')$ is a descendant of $(i_l, i_r)$.

As in~\cite{HT1984} we partition $T_S$ into disjoint paths. We classify each arc as either \emph{heavy} or \emph{light}. The root is light. For each internal arc $(i_l, i_r)$ we pick a child $(i^h_l, i^h_r)$ of maximum size and classify it as heavy. The remaining children are light. An edge to a light child is a \emph{light edge} and an edge to a heavy child is a \emph{heavy edge}. Let $\lightdepth(i_l, i_r)$ denote the number of light edges on the path from $(i_l, i_r)$ to the root of $T_S$. If $(i_l', i_r')$ is a light child of $(i_l, i_r)$, then $\size(i_l', i_r') \leq \size(i_l, i_r)/2$ since otherwise $(i_l', i_r')$ would be heavy. Consequently, the number of light edges on a path from the root to a leaf is at most logarithmic. Specifically, we will use the following well-known bound for trees restated for nested arc-annotated sequences.
\begin{lemma}[Harel and Tarjan \cite{HT1984}]\label{lem:lightdepth}
Let $S$ be a nested arc-annotated string containing the arc $(1, |S|)$. For any arc $(i_l, i_r) \in A_S$, $\lightdepth(i_l, i_r) \leq \log |A_S| + O(1)$.
\end{lemma}
Removing the light edges we partition $T_S$ into \emph{heavy paths}.

\subsection{Manipulating $\Gamma$ Sequences}\label{sec:gammasequence}
For positions $i_1$ and $i_2$ in $Q$, $i_1 \leq i_2$, define the \emph{$\Gamma$ sequence} for $i_1$ and $i_2$ as  
\begin{equation*}
\Gamma(i_1, i_2) = \gamma(m, m, i_1, i_2), \gamma(m-1, m, i_1, i_2), \ldots, \gamma(1, m, i_1, i_2).
\end{equation*}
Thus, $\Gamma(i_1, i_2)$ is the sequence of endpoints of the longest prefixes of each suffix of $P$ that is an arc-preserving subsequence of $Q[i_1, i_2]$. We can efficiently  manipulate $\Gamma$ sequences as suggested by the following lemma.
\begin{lemma}\label{lem:primitives}
For any positions $i_1$ and $i_2$ in $Q$, $i_1 \leq i_2$, we can compute in $O(m)$ time 
\begin{itemize}
\item[(i)] $\Gamma(i_2, i_2)$.
\item[(ii)] $\Gamma(i_1, i_2)$ from $\Gamma(i_1 + 1, i_2)$ if $(i_1, i_r) \not\in A_Q$ for every $i_r \leq i_2$.
\item[(iii)] $\Gamma(i_1, i_2)$ from $\Gamma(i_1, i_r)$ and $\Gamma(i_r + 1, i_2)$ if $(i_1, i_r) \in A_Q$ for some $i_r < i_2$.
\item[(iv)] $\Gamma(i_1,i_2)$ from $\Gamma(i_1, i_2 - 1)$, $\Gamma(i_1 + 1, i_2)$, and $\Gamma(i_1+1,i_2-1)$ if $(i_1, i_2) \in A_Q$. 
\end{itemize}
\end{lemma}
\begin{proof}
All the cases follow directly from the dynamic programming recurrence. Case (i) follows from case (2) of the recurrence, Case (ii) from  case (3) and (4) of the recurrence, Case (iii) from case (5) of the recurrence and Case (iv) from case (6)--(8) of the recurrence.
\end{proof}
We will use each of 4 cases in Lemma~\ref{lem:primitives} as primitive operations in our algorithm and we refer to (i), (ii), (iii), and (iv) as an \emph{initialize}, an \emph{extend}, a \emph{combine}, and a \emph{meld} operation, respectively. Fig.~\ref{fig:primitives} illustrates the extend, combine, and meld operations. An extend operation from $\Gamma(i_1 + k, i_2)$ to $\Gamma(i_1, i_2)$, for some $k > 1$, is defined to be the sequence of $k$ extend operations needed to compute $\Gamma(i_1, i_2)$ from $\Gamma(i_1 + k, i_2)$.
\begin{figure}[t] 
  \centering \includegraphics[scale=.5]{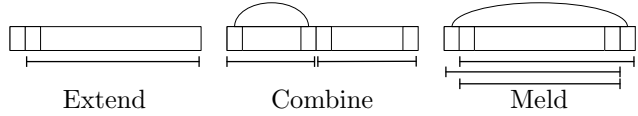}
  \caption{The extend, combine, and meld operations, respectively. For each operation the substring range(s) below the  string indicate the endpoints of the input $\Gamma$ sequence(s) needed in the operation to compute the $\Gamma$ sequence for the entire string.}
\label{fig:primitives}
\end{figure}

\subsection{The Algorithm}\label{sec:subalgorithm}
We now present our main algorithm. Initially, we construct $T_Q$ with a heavy path decomposition in $O(n)$ time and space. Then, we recursively compute $\Gamma$ sequences for each arc $(i_l, i_r) \in A_Q$ in a top-down traversal of $T_Q$. The $\Gamma$ sequence for the root contains the value $\gamma(1, m, 1, n)$ and hence this suffices to solve NAPS. The main idea is as follows. At each internal arc we first recursively compute the $\Gamma$ sequence for the heavy child and then compute $\Gamma$ sequences for the remaining light children in a right-to-left order (we will later see that this processing order is essential for the achieving the space bound). At the same time we use the extend and combine operations to compute $\Gamma$ sequences with a right endpoint at $i_r$ or $i_r - 1$ and a left endpoint at positions between children in the same left-to-right order. Finally, we use the meld operation on $\Gamma(i_l + 1, i_r)$ and $\Gamma(i_l, i_r -1)$ to get $\Gamma(i_l, i_r)$. 

At an arc $(i_l, i_r) \in A_Q$ in the traversal there are two cases to consider:

\paragraph{Case 1: $(i_l, i_r)$ is a leaf arc.} 
We compute $\Gamma(i_l, i_r)$ as follows.
\begin{enumerate}
\item Initialize  $\Gamma(i_r, i_r)$ and $\Gamma(i_r - 1, i_r - 1)$.
\item Extend $\Gamma(i_r , i_r)$ and $\Gamma(i_r - 1, i_r - 1)$ to get $\Gamma(i_l + 1 , i_r)$, $\Gamma(i_l, i_r - 1)$, and $\Gamma(i_l+1, i_r - 1)$.
\item Meld $\Gamma(i_l + 1 , i_r)$, $\Gamma(i_l, i_r - 1)$, and $\Gamma(i_l+1, i_r - 1)$ to get $\Gamma(i_l, i_r)$.
\end{enumerate}

\paragraph{Case 2: $(i_l, i_r)$ is an internal arc.}
Let $(i^1_l, i^1_r), \ldots, (i^s_l, i^s_r)$ be the childen arcs of $(i_l, i_r)$ in left-to-right order. To simplify the algorithm we set $i^0_r = i_l$. We compute $\Gamma(i_l, i_r)$ as follows.
\begin{enumerate}
\item Recursively compute $R_h := \Gamma(i^h_l, i^h_r)$, where $(i^h_l, i^h_r)$ is the heavy child arc of $(i_l, i_r)$. 
\item Initialize $\Gamma(i_r, i_r)$ and $\Gamma(i_r - 1, i_r -1)$. 
\item Extend $\Gamma(i_r, i_r)$ and $\Gamma(i_r - 1, i_r -1)$ to get $\Gamma(i^s_r + 1, i_r)$ and $\Gamma(i^s_r + 1, i_r - 1)$.
\item For $k := s$ down to $1$ do:
\begin{enumerate}
\item If $k \neq h$ recursively compute $R_k := \Gamma(i^k_l, i^k_r)$.
\item Combine $R_k$ with $\Gamma(i^k_r + 1, i_r)$ and with $\Gamma(i^{k}_r + 1, i_r - 1)$ to get $\Gamma(i^k_l, i_r)$ and $\Gamma(i^k_l, i_r - 1)$. 
\item Extend $\Gamma(i^k_l, i_r)$ and $\Gamma(i^k_l, i_r - 1)$ to get $\Gamma(i^{k-1}_r + 1, i_r)$ and $\Gamma(i^{k-1}_r + 1, i_r - 1)$. 
 \end{enumerate}
\item Extend  $\Gamma(i_l+1, i_r -1)$ to get $\Gamma(i_l, i_r -1)$.
\item Meld $\Gamma(i_l + 1, i_r)$, $\Gamma(i_l, i_r -1)$, and $\Gamma(i_l+1, i_r -1)$ to get $\Gamma(i_l, i_r)$. 
\end{enumerate}
\medskip
The computation in case 2 is illustrated in Fig.~\ref{fig:algorithm}. Note that when $k = 1$ in the loop in line 4, line 4(c) computes $\Gamma(i^{0}_r + 1, i_r) = \Gamma(i_l + 1, i_r)$ and $\Gamma(i_r^0 + 1, i_r - 1) =\Gamma(i_l+1, i_r - 1)$. In both cases above the algorithm computes several \emph{local $\Gamma$ sequences} of the form $\Gamma(i, i_r)$ and $\Gamma(i, i_r - 1)$, for some $i \leq i_r$. These sequences are computed in order of decreasing values of $i$ and each sequence only depends on the previous one and recursively computed $\Gamma$ sequences. Hence, we only need to store a constant number of local sequences during the computation at $(i_l, i_r)$. 
\begin{figure}[t] 
  \centering \includegraphics[scale=.5]{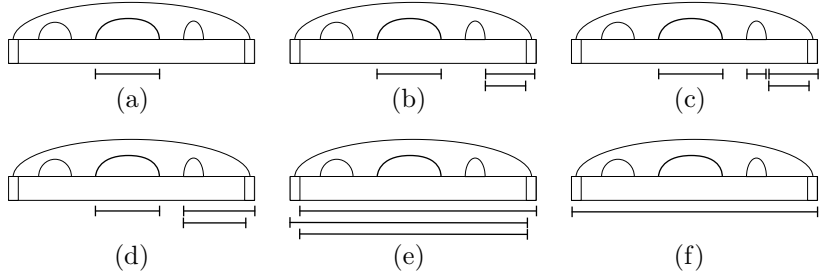}
  \caption{Snapshot of the $\Gamma$ sequences computed at an internal arc. The ranges below the arc-annotated sequences represent $\Gamma$ sequence endpoints. (a) After the recursive call to the heavy child in line 1. (b) After the extend operations in line 3. (c) After the recursive call in line 4(a) (d) After the combine operations in line 4(b). (e) Before the meld operation in line 6. (f) After the meld operation.}
\label{fig:algorithm}
\end{figure}

\subsection{Analysis}
We first consider the time complexity of the algorithm. To do so we bound the total number of primitive operations. For each arc in $A_Q$ there is $1$ initialize and $1$ meld operation and for each internal arc there is $1$ combine operation. Hence, the total number of initialize, meld, and combine operations is $O(|A_Q|)$. To count the number of extend operations we first define for any arc $(i_l, i_r) \in A_Q$ the set $\spaces(i_l, i_r)$ as the set of positions inside $(i_l, i_r)$ but not inside any child arc of $(i_l, i_r)$, that is, 
$$
\spaces(i_l, i_r) = \{i \mid i_l \leq i \leq i_r \text{ but not } i_l^k \leq i \leq i_r^k \text{ for any child  $(i_l^k, i_r^k)$ of $(i_l, i_r)$}\}.
$$
For example, $\spaces(1, 11)$ for $Q$ in Fig.~\ref{fig:treestructure}(a) is $\{1, 2, 11\}$. 
The $\spaces$ sets for all arcs is a partition of the positions in $Q$ and 
thus  $\sum_{(i_l, i_r) \in A_Q} \spaces(i_l, i_r) = n$. At an arc $(i_l, i_r)$ 
the algorithm performs 
$O(\spaces(i_l, i_r))$ extend operations and hence the total number of extend operations is $O(n)$. 
By Lemma~\ref{lem:primitives} each primitive operation takes $O(m)$ time and therefore the total running time of the algorithm is $O(|A_Q| m + n m) = O(nm)$.

For  the space complexity we bound the number of $\Gamma$ sequences stored by the algorithm. When the algorithm visits an arc $(i_l, i_r)$ we are currently processing a nested sequence of recursive calls corresponding to a path $p$ in $T_Q$ from the root to $(i_l, i_r)$. The number of $\Gamma$ sequences stored at each of these recursive calls is the total number of $\Gamma$ sequences stored. Consider an edge $e$ in $p$ from a parent $(i_l', i_r')$ to a child  $(i_l'', i_r'')$. If $e$ is heavy the recursive call to $(i_l'', i_r'')$ is done in line 1 of case 2 in the algorithm immediately at the start of the visit to $(i_l', i_r')$. Therefore, no $\Gamma$ sequence at $(i_l', i_r')$ is stored. If $e$ is light the recursive call to $(i_l'', i_r'')$ is done in line 4(a). The algorithm stores at most $3$ $\Gamma$ sequences, namely $\Gamma(i_r'' + 1, i_r')$, $\Gamma(i_r'' + 1, i_r' - 1)$, and $\Gamma({i^h_l}', {i^h_r}')$, where $({i^h_l}', {i^h_r}')$ is the heavy child of $(i_l', i_r')$. By Lemma~\ref{lem:lightdepth} there are at most $\log |A_Q| + O(1)$ light ancestors of $(i_l, i_r)$ in $T_Q$ and therefore the total space for stored $\Gamma$ sequences is $O(m \log |A_Q|)$. The additional space used by the algorithm is $O(n)$. Hence, we have the following result.

\begin{lemma}\label{lem:uncompressed}
Given nested arc-annotated strings $P$ and $Q$ of lengths $m$ and $n$, respectively, we can solve the nested arc-preserving subsequence problem in time $O(nm)$ and space $O(m\log |A_Q| + n)$.
\end{lemma}

\section{Squeezing into Linear Space}\label{sec:linearspace}
We now show how to compress $\Gamma$ sequences into a compact representation using $O(m)$ \emph{bits}. Plugging the new representation into our algorithm the total space becomes $O(n + m)$ as desired for Theorem~\ref{thm:main}. 

Our compression scheme for $\Gamma$ sequences relies on the following key property of the values of $\gamma$.
\begin{lemma}\label{lem:consecutive}
For any integers $j_1, j_2, i_1, i_2$, $1 \leq j_1 \leq j_2 \leq m$, $1 \leq i_1 \leq i_2 \leq n$, 
\begin{equation*}
j_1 -1\leq \gamma(j_1, j_2, i_1, i_2) \leq \gamma(j_1 + 1, j_2, i_1, i_2) \leq m
\end{equation*}
\end{lemma}
\begin{proof}
Adding another base in front of the substring $P[j_1 + 1, j_2]$ cannot increase the endpoint of an embedding of $P[j_1 + 1, j_2]$ in $Q$ and therefore $\gamma(j_1, j_2, i_1, i_2) \leq \gamma(j_1 + 1, j_2, i_1, i_2)$. Furthermore, for any substring $P[j_1, j_2]$ we can embed at most $j_2 - j_1$ bases and at least $0$ bases in $Q$ implying the remaining inequalities.
\end{proof}
Let $i_1, i_2$ be indices in $Q$ such that $i_1 \leq i_2$ and consider the sequence
\begin{equation*}
\Gamma(i_1, i_2) = \gamma(m, m, i_1, i_2), \ldots, \gamma(1,m, i_1, i_2) = \gamma_m, \ldots, \gamma_1
\end{equation*}
By Lemma~\ref{lem:consecutive} we have that $\gamma_m, \ldots, \gamma_1$ is a non-increasing and non-negative sequence where $\gamma_m$ is either $m$ or $m-1$. We encode the sequence efficiently using two bit strings $V$ and $U$ defined as follows. The string $V$ is formed by the concatenation of $m$ bit strings $s_m, \ldots, s_1$, that is, $V  = s_m \cdot s_{m-1} \cdots s_1$, where $\cdot$ denotes concatenation.  The string $s_m$ is the single bit $s_m = m - \gamma_m$ and $s_k$, $1 \leq k < m$, is given by 
\begin{equation*}
s_k = \begin{cases}
    0 & \text{if $\gamma_{k+1} - \gamma_{k} = 0$} \\
   \underbrace{1 \cdots 1}_{\gamma_{k+1} - \gamma_{k} \text{ times}} & \text{if $\gamma_{k+1} - \gamma_{k}  > 0$}
  \end{cases}
\end{equation*}
Let $D_k$ denote the sum of bits in string $s_m \cdots s_k$. We have that $m - D_m = m - s_m = \gamma_m$ and inductively $m - D_k = \gamma_k$. The string $U$ is the bit string of length $|V|$ consisting of a $1$ in each position where a substring in $V$ ends. Given $V$ and $U$ we can therefore uniquely recover $\gamma_m, \ldots, \gamma_1$. Since $\gamma_m, \ldots, \gamma_1$ can decrease by at most $m+1$ the total number of $1$s in $V$ is at most $m+1$. The total number of $0$s is at most $m$ and therefore $|V| \leq 2m+1$. Hence, our representation uses $O(m)$ bits. We can compress $\gamma_m, \ldots, \gamma_1$ into $V$ and $U$ in  a single scan in $O(m)$ time. Reversing the process we can also decompress in $O(m)$ time. Hence, we have the following result. 
\begin{lemma}\label{lem:compression}
We represent any $\Gamma$ sequence using $O(m)$ bits. Compression and decompression takes $O(m)$ time. 
\end{lemma}
We modify our algorithm from Section~\ref{sec:algorithm} to take advantage of Lemma~\ref{lem:compression}. Let $(i_l, i_r)$ be an internal arc in $A_Q$. Immediately before a recursive call to a light child $(i^k_l, i^k_r)$ of $(i_l, i_r)$ we compress the at most 3 $\Gamma$ sequences maintained at $(i_l, i_r)$, namely $\Gamma(i^h_l, i^h_r)$, where $(i^h_l, i^h_r)$ is the heavy child, $\Gamma(i^k_r + 1, i_r)$, and  $\Gamma(i^k_r + 1, i_r - 1)$. Immediately after returning from the recursive call we decompress the sequences again. 

The total number of compressions and decompressions is $O(n)$. Hence, by Lemma~\ref{lem:compression} the additional time used is $O(nm)$ and therefore the total running time of the algorithm remains $O(nm)$. The space for storing the $O(\log |A_Q|)$ $\Gamma$ sequences becomes $O(m \log |A_Q|) =  O(m \log n)$ \emph{bits}. Hence, the total space is $O(n +m)$. In conclusion, we have shown Theorem~\ref{thm:main}.

\subsection{Avoiding Decompression}				
The above algorithm requires $O(n)$ decompressions. We briefly describe how one can avoid these decompressions by augmenting the representation of $\Gamma$ sequences slightly. A \emph{rank/select index} for a bit string $B$ supports the operations $\rank(B, k)$ that returns the number of $1$s in $B[1,k]$ and $\select(B, k)$ that returns the position of the $k$th $1$ in $B$. We can construct a rank/select index in $O(|B|)$ time that uses $o(|B|)$ bits and supports both operations in constant time~\cite{Munro1996}. We add a rank/select index to the bit strings $V$ and $U$ in our compressed representation. Since these use $o(m)$ bits this does not affect the space complexity. Let $\gamma_m, \ldots, \gamma_1$ be a $\Gamma$ sequence compressed into bit strings $V$ and $U$ augmented with a rank/select index. For any $k$, $1\leq k\leq m$ we can compute the element $\gamma_k$ in constant time as
\begin{equation*}
m - \rank(V, \select(U, m+1-k))
\end{equation*}
To see the correctness, first note that $\select(U,m+1-k)$ is end position of the $m+1-k$th substring in $V$. Therefore, $\rank(V, \select(U, m+1-k))$ is the sum of the bits in the first $m+1-k$ substrings of $V$. This is $D_k$ and since $\gamma_k = m -D_k$ the computation returns $\gamma_k$. In summary, we have the following result.
\begin{lemma}
We can represent any $\Gamma$ sequence in $O(m)$ bits while allowing constant time access to any element. 
\end{lemma}
The algorithm now only needs to compress $\Gamma$ sequences once. Whenever, we need an element of a compressed $\Gamma$ sequence we extract it in constant time as above. Hence, the asymptotic complexity of the algorithm remains the same.

\bibliographystyle{abbrv}
\bibliography{paper}

\end{document}